\documentclass[letterpaper, 10 pt, conference]{ieeeconf}  %

\IEEEoverridecommandlockouts                              %

\usepackage{cite}
\usepackage{amsmath,amssymb,amsfonts}
\usepackage{graphicx}
\usepackage{textcomp}
\usepackage{xcolor}

\usepackage{subfigure}

\usepackage{kotex}
\usepackage{epsfig}

\usepackage{amsthm}
\usepackage{multirow}

\title{\LARGE \bf
CNN-based End-to-End Adaptive Controller with Stability Guarantees
}

\author{Myeongseok Ryu$^{1}$ and Kyunghwan Choi$^{1}$%
\thanks{*This work was supported by Korea Institute for Advancement of Technology (KIAT) grant funded by the Korea Government (MOTIE) (P0020535, The Competency Development Program for Industry Specialist)}%
\thanks{$^{1}$Myeongseok Ryu and Kyunghwan Choi are with the School of Mechanical Engineering, 
        Gwangju Institute of Science and Technology, Gwangju, 61005, Republic of Korea
        {\tt\small dding\_98@gm.gist.ac.kr, khchoi@gist.ac.kr}}%
}

\begin{document}

\newtheorem{theorem}{Theorem}
\newtheorem{remark}{Remark}
\newtheorem{assum}{Assmption}

\maketitle
\thispagestyle{empty}
\pagestyle{empty}

\begin{abstract}

This letter proposes a convolutional neural network (CNN)-based adaptive controller wtih three notable features: 1) it determines control input directly from historical sensor data (in an end-to-end process); 2) it learns the desired control policy during real-time implementation without using a pretrained network (in an online adaptive manner); and 3) the asymptotic tracking error convergence is proven during the learning process (to deliver a stability guarantee). An adaptive law for learning the desired control policy is derived using the gradient descent optimization method, and its stability is analyzed based on the Lyapunov approach. A simulation study using a control-affine nonlinear system demonstrated that the proposed controller exhibits these features, and its performance can be tuned by manipulating the design parameters. In addition, it is shown that the proposed controller has a superior tracking performance to that of a deep neural network (DNN)-based adaptive controller.

\end{abstract}

\section{Introduction}

Neural networks (NNs) have been widely used in control applications as a function approximator for adaptive control \cite{BookEKcontrol}, state estimation \cite{BoockEKestimation}, and so on. 
These studies provided the ultimate boundedness of tracking or estimation errors based on the Lyapunov-based stability analysis, but were limited to investigating shallow NNs with one hidden layer.

Deep neural networks (DNNs) with multiple hidden layers are more expressive and can provide better performance than shallow NNs. However, due to the difficulty in deriving an adaptation law ensuring stability, designing a stable DNN-based controller is generally considered challenging. There have been a few attempts at overcoming this challenge, including developing Lyapunov-based adaptation laws for controllers based on DNNs or theirs variations \cite{DixonDNN,DixonLSTM,DixonPINN}. A Lyapunov-based adaptation law for a DNN-based controller was derived in \cite{DixonDNN}, which ensured asymptotic convergence of tracking error. \cite{DixonLSTM} utilized long short-term memory (LSTM), a type of DNN, in designing an adaptive controller and presented the ultimate boundedness of tracking error based on Lyapunov analysis.
The Lyapunov-based approach was extended to using physics-informed LSTM for an adaptive controller and therein demonstrated its asymptotic stability \cite{DixonPINN}.

While DNN-based adaptive control with a stability guarantee has been studied by pioneering works \cite{DixonDNN,DixonLSTM,DixonPINN}, convolutional neural networks (CNNs), which are well known for their spatial feature capturing ability, have not been as actively investigated in control applications, as they have been in computer vision applications.
Nonetheless, motivated by the feature-capturing capability, the use of CNNs as a basis for end-to-end controllers has been studied \cite{CNNimg2strOutErr,CNNimg2strAngErr,CNNmat2mat,CNNsensorSystemCompare}. These end-to-end controllers determine the control input directly from the sensor data with features that are extracted by CNNs. In \cite{CNNimg2strOutErr,CNNimg2strAngErr}, CNN-based end-to-end controllers were trained offline to produce the desired steering behavior based on 2D images from the camera. \cite{CNNmat2mat,CNNsensorSystemCompare} generated pseudo 2D images by stacking historical system input or output data and used them to train CNN-based end-to-end controllers offline to behave as target controllers.
However, none of the works above considered online adaptation or stability analysis of CNN-based end-to-end controllers. 

This letter presents a CNN-based end-to-end adaptive controller with a stability guarantee.
The proposed controller uses 2D images generated by stacking historical sensory data as an input matrix to convolutional layers (CVLs). The output of CVLs is input to the following fully-connected layers (FCLs), which produce the controller output. An adaptation law of network weights are implemented to learn the desired control policy and is derived using the gradient descent optimization method. 
The stability of the adaptation laws is proven based on the Lyapunov analysis by showing that the tracking error is asymptotically convergent and the network weights and biases are bounded. 
A part of the proposed controller is designed based on these findings, such as Jacobians with respect to FCLs from \cite{DixonDNN} and stabilizing techniques for the adaptation law from \cite{BoockEKestimation,BookEKcontrol}; however, the remainder of the proposed controller, particullary the components rely on the mathematical formulation and adaptation law of the overall CNN architecture, has been solely developed by the present authors.
A simulation study using a control-affine nonlinear system is performed to demonstrate that the proposed CNN-based end-to-end controller ensures asymptotic convergence of tracking error during online adaptation without using a pretrained network. 
In addition, the tracking performance of the proposed controller is compared with that of the DNN-based adaptive controller, which is a form in which the CVLs are not included in the proposed controller.
\begin{figure*}[!t]
    \centering
    \includegraphics[width=2.2\columnwidth]{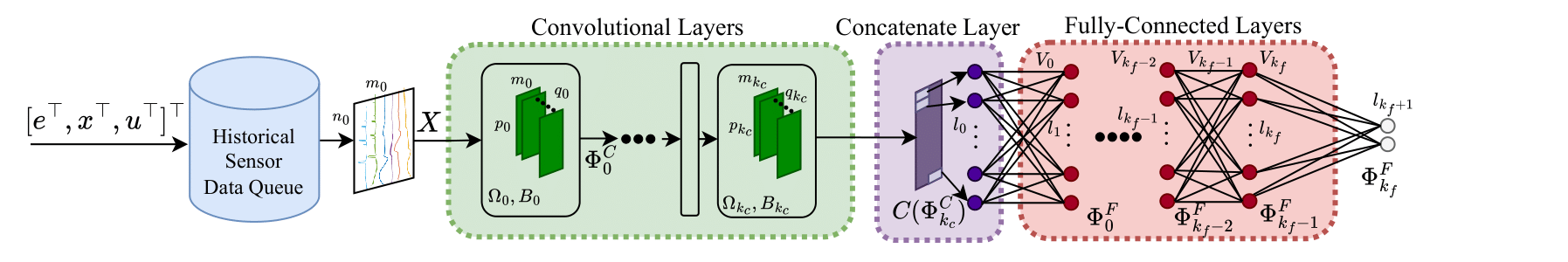}
    \caption{CNN architecture used in the proposed controller. 
    } 
    \label{fig: CNN structure}
\end{figure*}
\section{Problem Formulation}

In this section, the following notation is defined and then the problem formulation is presented.

\begin{itemize}
    \item $\odot$ and $\otimes$ denote the Hadamard and Kronecker products, respectively \cite{BookMatrix}.
    \item $x_{(i)}$ denotes the $i\textsuperscript{th}$ element of vector $x$, and $x_{(i:j)}\triangleq [x_{(i)},x_{(i+1)},\dots,x_{(j)}]$ where $i \le j$.    
    \item For $x\in\mathbb{R}^{pq}$ and $pq=nm$, $\text{reshape}(x, n, m) \triangleq[x_{(1)},\dots,x_{(n)};x_{(n+1)},\cdots,x_{(2n)};x_{(2n+1)},\cdots,x_{(nm)}]^\top$.    
    \item $A_{(i,j)}$ is the $i\textsuperscript{th}$ row-$j\textsuperscript{th}$ column element of matrix $A$.
    \item $\text{vec}(A)\triangleq [A_{(1,1)},\dots,A_{(1,m)},A_{(2,1)},\cdots,A_{(n,m)}]^\top  $ for $A\in\mathbb{R}^{n\times m}$.
    \item $\text{row}_i(A) \triangleq [A_{(i,1)},A_{(i,2)},\cdots,A_{(i,m)}]$ denotes the $i\textsuperscript{th}$ row of $A\in\mathbb{R}^{n\times m}$.
    \item $\text{row}_{(i:j)} (A) \triangleq [\text{row}_i(A)^\top  , \text{row}_{i+1}(A)^\top  ,\dots,\text{row}_{j}(A)^\top  ]^\top  $ denotes a matrix consisting of rows $i$ to $j$ of $A\in\mathbb{R}^{n\times m}$.
    \item $\overset{\curvearrowleft}\prod^{n}_{l=k} A_l \triangleq \left\{ {\begin{array}{*{20}{c}}
{{A_n}{A_{n - 1}} \cdots {A_k},\,\,{\rm{if}}\,\,n \ge k}\\
{1,\,\,\,\,\,\,\,\,\,\,\,\,\,\,\,\,\,\,\,\,\,\,\,\,\,\,\;\;\;\;\;\;\;{\rm{if}}\,\,n < k}
\end{array}} \right.$, where $k$ and $n$ are positive integers and $A_l$ is a matrix.
\end{itemize}

Consider a control-affine nonlinear system modeled as 
\begin{equation}
    \dot x = f(x) + u,
    \label{eq: model dynamics}
\end{equation}
where $x\in \mathbb{R}^n$ denotes the state, $u\in\mathbb{R}^n$ denotes the control input, and $f:\mathbb{R}^n\to\mathbb{R}^n$ denotes an unknown smooth nonlinear function. System \eqref{eq: model dynamics} can be reformulated as 
\begin{equation}
    \dot x = A_c x + f_c(x) + u,
    \label{eq: reformed dynamics}
\end{equation}
where $f_c(x)= f(x) - A_c x$ and $A_c\in\mathbb{R}^{n\times n}$ is a Hurwitz designer matrix. 

The desired control policy $u^*$ is defined by feedback linearization control with the knowledge of function $f(x)$ as follows:
\begin{equation}
    u^* = -f_c(x) + \dot x_d - A_c x_d - k_s \text{sgn}(e),
    \label{eq: desired input}
\end{equation}
where $x_d\in\mathbb{R}^n$ is the desired state trajectory, which is continuously differentiable; $e\triangleq x-x_d$ is the tracking error; and $k_s\in\mathbb{R}_{>0}$ is the control gain.
The desired control policy provides the following desired error dynamics when applied to the system \eqref{eq: reformed dynamics}: 
\begin{equation}\label{eq: ideal error dynamics}
    \dot e = A_c e-k_s\text{sgn}(e),
\end{equation}
which provides asymptotic convergence of the tracking error. The desired control policy is an ideal framework, but its solution unknown due to the use of the unknown function $f(x)$. The goal is to design a controller that can learn this desired control policy online with a guarantee of stability.

\section{Controller Design}\label{label: CNN declare}

This section presents a CNN-based adaptive controller. 
The CNN architecture $\Phi:\mathbb{R}^{n_0\times m_0}\times \mathbb{R}^{p_0\times m_0\times q_0}\times \mathbb{R}^{q_0} \times \cdots \times \mathbb{R}^{p_{k_c}\times m_{k_c}\times q_{k_c}}\times \mathbb{R}^{q_{k_c}}\times \mathbb{R}^{(l_0+1)\times l_1}\times \cdots \times \mathbb{R}^{(l_{k_f}+1)\times l_{k_f+1}} \to \mathbb{R}^{l_{k_f+1}}$ illustrated in Fig.~\ref{fig: CNN structure} consists of the CVLs, concatenate layer and FCLs.

\subsection{Convolutional Layers}
The CVLs are represented recursively as
\begin{equation}
    \Phi^C_{j_c} = 
    \begin{cases}
        O(\phi^C_{j_c}(\Phi^C_{j_c-1}),\Omega_{j_c},B_{j_c}),&
        j_c\in[1,\dots,k_c]\\
        O(\phi^C_0(X),\Omega_0,B_0),&j_c=0
    \end{cases},
\end{equation}
where $X\in\mathbb{R}^{n_0\times m_0}$ denotes the network input matrix, $O:\mathbb{R}^{n_{{j_c}}\times {m_{j_c}}} \times \mathbb{R}^{p_{j_c}\times m_{j_c} \times q_{j_c} }\times \mathbb{R}^{q_{j_c}}\to \mathbb{R}^{n_{j_c+1}\times m_{j_c+1}}$ denotes the CNN operator (see Appendix A), and $\phi^C_{j_c}:  \mathbb{R}^{n_{j_c}\times m_{j_c}} \to  \mathbb{R}^{n_{j_c}\times m_{j_c}}$ denotes the matrix activation function (i.e. $ {\phi^C_{j_c}}(\Phi^C_{j_c-1})_{(i,j)}= \sigma_{j_c}({\Phi^C_{j_c-1}}_{(i,j)}))$ for some activation functions $\sigma_{j_c}:\mathbb{R}\to\mathbb{R}$ .
The first activation function $\phi^C_0$ should be a bounded nonlinear function to ensure that the input to the first CNN operator is bounded. In this study, $\alpha_1\tanh(\cdot)$ is selected with $\alpha_1\in\mathbb{R}_{>0}$.
The filter set $\Omega_{j_c}$ contains $q_{j_c}$ filters $W^{(i)}_{j_c}\in\mathbb{R}^{p_{j_c}\times m_{j_c}},\ \forall i\in[1,\dots, q_{j_c}]$ and is represented as $\Omega_{j_c} = \{W^{(1)}_{j_c},\dots,W^{(q_{j_c})}_{j_c}\}$, where superscript $i$ denotes the filter index. The bias vector $B_{j_c}$ consists of $q_{j_c}$ biases. 

\subsection{Fully-Connected Layers}

The output matrix of the CVLs is input to the concatenate layer $C(\Phi^C_{k_c}) = [\text{vec}({\Phi^C_{k_c}}^\top  )^\top  ,1]^\top$ before the input layer of FCLs for compatibility between the CVLs and FCLs.

The FCLs are represented recursively as 
\begin{equation}
    \Phi^F_{j_f} =
        \begin{cases}
            V^\top  _{j_f}\phi^F_{j_f}(\Phi^F_{j_f-1}),  & j_f\in[1,\dots,k_f]   \\
            V_0^\top   C(\Phi^C_{k_c}),                  & j_f=0
        \end{cases},
\end{equation}
where $V_{j_f}\in\mathbb{R}^{l_{j_f}+1\times l_{j_f+1}}$ denotes the weight matrix, $\phi^F_{j_f}:\mathbb{R}^{l_{j_f}}\to\mathbb{R}^{l_{j_f+1}}$ denotes the vector activation function defined by
$\phi^F_{j_f} \triangleq \phi^F_{j_f}(\Phi^F_{j_f-1})=[\sigma_{j_f}({\Phi^F_{j_f-1}}_{(1)}),\sigma_{j_f}({\Phi^F_{j_f-1}}_{(2)})\cdots,\sigma_{j_f}({\Phi^F_{j_f-1}}_{(l_{j_f})}),1]^\top$ for some nonlinear activation functions $\sigma_{j_f}:\mathbb{R}\to\mathbb{R}$ such as $\tanh(\cdot)$. Note that $C(\Phi^C_{k_c})$ and $\phi^F_{j_f}$ are augmented by 1 to consider the bias of the preceding input layer as a weight in the weight matrix. Hereafter, weights refer to the trainable variables in filters, bias vectors, and weight matrices for simplicity.

The output of the FCLs is the final output of the CNN architecture and is represented as
\begin{equation}
    \Phi(X,\Omega_0,B_0,\dots,\Omega_{k_c},B_{k_c},V_0,\cdots,V_{k_f}) \triangleq \Phi^F_{k_f}.
\end{equation}

\subsection{Control Law Development} \label{Coltrol Law Development}

In the controller, the CNN architecture is designed to approximate the lumped term 
$\Lambda\triangleq f_c(x)-\dot x_d+A_cx_d$ in the desired control policy \eqref{eq: desired input}.
According to the universal approximation theorem \cite{UnivCNN}, there exist ideal bounded network weights
$\Omega^*_{j_c},B^*_{j_c},\ \forall j_{c}\in [0,\dots,k_c]$ and $V^*_{j_f},\ \forall j_f\in [0,\dots,k_f]$, 
such that $||\Lambda- \Phi^*|| \le \bar \epsilon$
where $\Phi^*\triangleq \Phi(X,\Omega_0^*,B_0^*,V_0^*,\dots,\Omega_{k_c}^*,B_{k_c}^*,V_{k_f}^*)$ for a positive constant $\bar\epsilon$. In other words, there exist some positive constants $\bar\Omega$, $\bar B$ and $\bar V$, such that $||\Omega^*_{j_c}||\le \bar\Omega$, $|| B^*_{j_c}||\le \bar B$ and $|| V^*_{j_f}||\le \bar V$ for all $j_c\in[0,\dots,k_c]$, $j_f\in[0,\dots,k_f]$.

Then, the lumped term $\Lambda$ can be approximated as follows:
\begin{align}        
    u^* = -\Phi^* -\epsilon -k_s\text{sgn}(e),
\end{align}
where $\epsilon$ is the approximation error such that $||\epsilon||\le\bar\epsilon$.  
The CNN architecture adapts its weights to converge to the ideal values while used in the control input as 
\begin{equation}
u = -\hat\Phi - k_s\text{sgn}(e), \label{eq: proposed ctrl}
\end{equation}
where $\hat\Phi\triangleq \Phi(X,\hat\Omega_0,\hat B_0,\hat V_0,\dots,\hat\Omega_{k_c},\hat B_{k_c},\hat V_{k_c})$ denotes the CNN architecture with weights being adapted.
The control input $u$ results in the following error dynamics:
\begin{equation}
    \dot e = A_c e-k_s\text{sgn}(e)+\tilde \Phi +\epsilon,
    \label{eq: error dynamics with NN error}
\end{equation}
where $\tilde\Phi\triangleq \Phi^*-\hat\Phi$. Note that the last two terms on the right side have been added compared to \eqref{eq: ideal error dynamics}.

\subsection{Network Input Matrix Design}

Any 2D image data can be used as the input matrix of this CNN architecture if the data include sufficient information for estimating the system dynamics. For instance, concatenated camera images and time-stacked system input and output data constitute sufficient datasets. This study adopts the latter, and the resulting input matrix is defined as
\begin{equation}
    X(t)\triangleq[\xi(t), \xi(t-T_s), \cdots, \xi(t-({n_0}-1) T_s)]^\top\in\mathbb{R}^{n_0\times m_0},
\end{equation}
where $\xi=\alpha_2[e^\top  ,x^\top  ,u^\top  ]^\top  \in\mathbb{R}^{2n+m}$, $\alpha_2$ denotes a positive constant, $T_s\in\mathbb{R}_{>0}$ denotes the data stacking time, and $n_0$ denotes the number of stacks.
Parameter $\alpha_2$ prevents $X$ from being saturated or distorted by the activation functions.
\section{Weight Adaptation Law}

Jacobians with respect to the CVLs and FCLs are derived first, followed by a derivation of the weight adaptation law based on the gradient descent optimization method using the Jacobians.

\subsection{Jacobians with respect to Convolutional Layers}\label{subsec: Jacobian CNN}

Let $\Phi_i$ denote the $i\textsuperscript{th}$ output of $\Phi$.
Then, the Jacobians of $\Phi_i$ with respect to the weights of CVLs are derived as
\begin{equation}
    \begin{aligned}
        {\partial \Phi_i\over\partial W^{l_k}_{j_c}} 
        &=
        \sum_{l_i=1}^{n_{j_c+1}}
        \bigg(
        {\partial \Phi_i\over\partial {\Phi^C_{j_c}}_{(l_i,l_k)}} 
            \text{row}_{(l_i:l_i+p_{j_c}-1)}(\phi^C_{j_c})
        \bigg)
        \\
         &\hspace{90pt}\text{with }j_c\in[0,\dots,k_c],\hfilneg\\        
        {\partial \Phi_i\over\partial {B_{j_c}}_{(l_k)}} &=
        \sum_{l_i=1}^{n_{j_c+1}} 
        \bigg(
            {\partial \Phi_i\over\partial {\Phi^C_{j_c}}_{(l_i,l_k)}} \cdot 1
        \bigg)
        \\
         &\hspace{90pt}\text{with }j_c\in[0,\dots,k_c],\hfilneg
    \end{aligned}
\end{equation}
for all $l_k\in[1,\cdots,q_{j_c}]$, $i\in[1,\dots,l_{k_f+1}]$, where $\phi^C_{j_c}\triangleq\phi^C_{j_c}(\Phi^C_{j_c-1})$ for $j_c\in[1,\dots,k_c]$, $\phi^C_0 \triangleq \phi^C_0(X)$, and $\partial \Phi_i/\partial \Phi^C_{j_c}$ denotes the backpropagated gradient of $\Phi_i$ with respect to $\Phi^C_{j_c}$, which is obtained using the back-propagation method. The details of the Jacobian calculation are provided in Appendix B.
\subsection{Jacobians with respect to Fully-connected Layers}\label{subsec: Jacobian FCN}

The Jacobians of $\Phi$ with respect to the weights of FCLs were derived in \cite{DixonDNN} as
\begin{equation}
    \begin{aligned}
    \frac{\partial \Phi}{\partial \text{vec}( V_0)} &= 
        (\overset{\curvearrowleft}\prod^{k_f}_{l=1}  V^\top  _l{\phi^{F}_l}') ( I_{l_{1}} \otimes C(\Phi^C_{k_c})),\\
    &\hspace{90pt}\text{with }{j_f} =0, \hfilneg
    \\
    {\partial \Phi\over\partial \text{vec}(V_{j_f})} &=  
        (\overset{\curvearrowleft}\prod^{k_f}_{l={j_F}+1} V^\top  _l{\phi^{F}_l}') ( I_{l_{{j_f}+1}} \otimes {\phi_{j_f}^F}^\top  ), \\
    &\hspace{60pt}\text{with } {j_f} \in[1,\dots,k_f], \hfilneg
    \end{aligned}
\end{equation}
where ${\phi^F_{j_f}}'\triangleq {\partial \over \partial x} {\phi^F_{j_f}}(x)$ is the Jacobian of the activation function with respect to some vector $x$. 

\subsection{Derivation of Adaptation Laws} 

For simplicity, the weights of the FCLs and CVLs are expressed as column vectors:
$\theta_F \triangleq [\text{vec}( V_{0});\cdots;\text{vec}( V_{k_f})]$ and $\theta_C \triangleq [\text{vec}(W^{(1)}_0);\text{vec}(W^{(2)}_0);\cdots;\text{vec}(W^{(q_{k_c})}_{k_c});B_0;\cdots;B_{k_c}]$, respectively. Define the weight vector including all the weights as $\theta\triangleq [\theta_F;\theta_C]\in\mathbb{R}^\Xi$, where $\Xi\triangleq\sum_{j_f=0}^{k_f} (l_{j_f}+1) l_{j_f+1} + {\sum_{j_c=0}^{k_c} (p_{j_c}m_{j_c}+1)q_{j_c}}$ denotes the length of $\theta$. 
Then, the Jacobian of $\Phi$ with respect to the weight vector is represented as
\begin{equation}
    \Phi'\triangleq 
    {\partial \Phi\over\partial \theta} = 
    \begin{bmatrix}
        \dfrac{\partial \Phi}{\partial \theta_F} & \dfrac{\partial \Phi}{\partial \theta_C}
    \end{bmatrix}
    \in\mathbb{R}^{l_{k_f+1}\times \Xi},
\end{equation}  
where
\begin{equation}
    \begin{aligned}
        {\partial \Phi\over\partial  \theta_F} = &
        \begin{bmatrix}
        \dfrac{\partial \Phi}{\partial \text{vec}( V_0)}
        & \cdots &
        \dfrac{\partial \Phi}{\partial \text{vec}( V_{k_f})}    
        \end{bmatrix}
        \\
        {\partial \Phi\over\partial \theta_C} = &
        \begin{bmatrix}
        \dfrac{\partial \Phi_1}{\partial  \theta_C}^\top  
        & \cdots &
        \dfrac{\partial \Phi_{l_{k_f+1}}}{\partial  \theta_C}^\top      
        \end{bmatrix}^\top.
    \end{aligned}
\end{equation}

Consider a convex objective function $J={1\over 2}e^\top  e$, which is defined to minimize the tracking error. The weight adaptation law is derived by employing the gradient descent optimization method to minimize $J$, motivated by \cite{BookEKcontrol}, \cite{BoockEKestimation}. The gradient of $J$ with respect to the $\hat\theta$ is 
\begin{equation} 
    {\partial J\over\partial \hat\theta} = 
    \bigg(
    {\partial J\over\partial e}{\partial e\over\partial \hat\theta}
    \bigg)^\top   
    =
    \bigg(
    e^\top  {\partial e\over\partial \hat\theta}
    \bigg)^\top  
    ,
\end{equation}
where $\hat\theta\triangleq [\hat\theta_F;\hat\theta_C]$ denotes the estimate of $\theta^*$. 
By applying a static approximation (i.e., $\dot e = 0$) to \eqref{eq: error dynamics with NN error} and
, the term $\partial e/\partial \hat\theta$ can be calculated as
\begin{equation}
    {\partial e\over\partial \hat\theta} = A_c^{-1} \hat\Phi',
\end{equation}
where $\hat\Phi'\triangleq \partial\hat\Phi/\partial\hat\theta$.
Finally, the weight adaptation law is proposed as 
\begin{equation}
        \dot{\hat \theta} =
            \text{proj}
            \bigg[
            -\Gamma 
            (A_c^{-1}\hat\Phi')^\top   e 
            -\rho||e|| \hat\theta
            \bigg]
        ,
    \label{eq: adaptation law}
\end{equation}
where $\Gamma\in\mathbb{R}^{\Xi\times \Xi}$ denotes the learning rate matrix which is symmetric and positive-definite and $\rho\in\mathbb{R}_{>0}$ is the damping factor which corresponds to the e-modification term presented in \cite{BookEKcontrol}, \cite{BoockEKestimation}. 

The estimated weight vector $\hat\theta$ remains in set $\Theta_{\theta}=\{\theta \ |\  ||\theta||\le\bar\theta\}$ (i.e., $\text{sup}_{\hat\theta\in\Theta_{\theta}} ||\hat\theta|| \le \bar \theta\in L_\infty$) by the projection operator (Appendix E, Eq.~(E.4) in \cite{BookProjection}).
As mentioned in Section \ref{Coltrol Law Development}, the ideal weights are bounded, it follows that ideal weight vector $\theta^*\triangleq[\theta^*_F;\theta_C^*]$ is also bounded by a positive constant $\bar\theta$ such that $||\theta^*||\le\bar\theta$. 
Therefore, weights estimate error $\tilde\theta\triangleq\theta^*-\hat\theta$ is also bounded.

Moreover, there exists a positive constant $\Phi'_M$ such that $||\hat\Phi'||\le\Phi'_M$, because the activation functions and their gradients are bounded for some bounded inputs, the input matrix $X$ and the weight estimates are bounded. 

\section{Stability Analysis}

The estimation error of the CNN architecture is expressed as follows using the first-order Taylor series approximation:
\begin{align}\nonumber
        \tilde \Phi =& 
            \Phi(X,\theta^*)-\Phi(X,\hat\theta)
        \\
            =&\hat\Phi'\tilde \theta+\mathcal{O}(||\tilde \theta||^2),
    \label{eq: tayor approx}
\end{align}
where $\mathcal{O}(\cdot)$ denotes a higher-order error.
The error dynamics \eqref{eq: error dynamics with NN error} are reexpressed as follows using the relationship \eqref{eq: tayor approx}:
\begin{equation}
    \dot e = A_c e - k_s\text{sgn}(e)+\hat\Phi'\tilde \theta +\Delta,
\end{equation}
where $\Delta \triangleq \epsilon + \mathcal{O}(||\tilde\theta||^2)\le \bar\Delta\in L_\infty$ denotes a lumped error. 
The following theorem establishes the asymptotic tracking error convergence of the proposed controller.
\begin{theorem}
For the dynamical system in \eqref{eq: model dynamics}, the proposed controller in \eqref{eq: proposed ctrl} and adaptation law \eqref{eq: adaptation law} ensure asymptotic tracking error convergence, in the sense that $e\to0$ as $t\to \infty$, provided that $\beta_1\beta_2^2+ \bar\Delta \le k_s$ where 
$\beta_1=\rho||\Gamma^{-1}||$ and $\beta_2=\{\Phi'_M(||A_c^{-1}||+1)+\beta_1\bar\theta\}/2\beta_1$.

\end{theorem}

\begin{proof}
Let $V:\mathbb{R}^{n+\Xi}\to\mathbb{R}_{>0}$ denote the candidate Lyapunov function defined as 
\begin{equation}
    V={1\over2}e^\top e+{1\over 2}\tilde \theta^\top  \Gamma^{-1}\tilde \theta.
\end{equation}
Using $\tilde x^T\hat x\le||\tilde x||(\bar  x - ||\tilde x||)$, $x^T\text{sgn}(x)=||x||_1\ge ||x||$ for $x\in\mathbb{R}^n$, \textit{Lemma E.1} in \cite{BookProjection}, and $\dot \theta^*=0$, the time-derivative of $V$ is derived as 
\begin{equation}
    \begin{aligned}
        \dot V 
        =
        &
        e^\top  A_c e + 
        e^\top   
        \bigg( -k_s\text{sgn}(e)+\hat\Phi'\tilde \theta+\Delta
        \bigg)
        \\
        &-\tilde\theta^\top   \Gamma^{-1} 
        \text{proj}
        \bigg[
            -\Gamma  \hat\Phi'^\top  {A_c^{-1}}^\top  e
            - \rho||e|| \hat \theta
        \bigg]
\\
\le &
        e^\top  A_c e -k_s ||e|| + ||e|| (\Phi'_M ||\tilde \theta||  +\bar\Delta ) 
\\      
        & + \Phi'_M ||\tilde \theta||   ||A_c^{-1}|| ||e|| 
        +   \rho||\Gamma^{-1}||||\tilde\theta|| ||e|| (\bar\theta - ||\tilde \theta||)
\\
\le&
        e^\top  A_c e 
        +||e|| \bigg\{ 
         -k_s+        \bar\Delta
\\
        &
        +||\tilde \theta|| \bigg(\Phi'_M (||A^{-1}||+1) + \rho||\Gamma^{-1}|| (\bar\theta - ||\tilde \theta||) \bigg)
        \bigg\}
\\
\le&
        e^\top  A_c e +
        ||e|| 
        \bigg\{
         -\beta_1(||\tilde \theta ||-\beta_2)^2 + \beta_1\beta_2^2 + \bar\Delta -k_s 
        \bigg\},
       \end{aligned}
       \label{eq: dot V}
\end{equation}
where $\beta_1\triangleq\rho||\Gamma^{-1}||$ and $\beta_2 \triangleq \{\Phi_M' (||A^{-1}||+1) + \beta_1\bar\theta\}/2\beta_1$ are positive constants.
By selecting $k_s$ such that $\beta_1\beta_2^2+\bar\Delta \le k_s $, \eqref{eq: dot V} yields 
\begin{equation}
    \begin{aligned}
        \dot V \le e^\top  A_c e  -\beta_1||e||(||\tilde\theta||-\beta_2)^2\le
        e^\top  A_c e \triangleq -W(e) 
    \end{aligned}.\label{eq: W}
\end{equation}
Inequality \eqref{eq: W} implies that $\int_{{t_0}}^t {W(e(\tau ))} \,d\tau $ is finite. By Barbalat's Lemma \cite{Khalil}, $W(e) \to 0$ as $t \to \infty$. Therefore, $e\to 0$ as $t\to\infty$. 
\end{proof}

\begin{remark}
The design parameters $n_0$ and $T_s$ of the input matrix $X$ play a crucial role in determining the resolution and size of the information provided to the CVLs. The smaller $T_s$ is, the higher the resolution becomes, while also limiting the temporal range of the information.
As $n_0$ increases, the input matrix possesses information for a longer time, which also results in higher computational costs.
\label{remark:stack param}
\end{remark}

\begin{remark}
As mentioned in Remark 3.3 of \cite{BookEKcontrol}, the stability and convergence speed in the learning phase are affected by design parameters $\Gamma, \rho,$ and $A_c$. 
Increasing the learning rate $\Gamma$ increases the convergence speed, but some oscillations can occur in the transient response. 
The damping factor $\rho$ can help the learning phase to be robust to NN approximation errors \cite{BookSSGe}, but the weights may converge far from the optimal weights. 
The further the $A_c$'s eigenvalues are from the imaginary axis in the negative direction, the faster the tracking error dynamics are, but the slower the weights are updated due to the use of $A_c^{-1}$ in the adaptation law \eqref{eq: adaptation law}.
\label{remark:2}
\end{remark}

\section{Simulations} \label{sec: sim}

Comparative simulations were performed to demonstrate the efficacy of the proposed CNN-based end-to-end adaptive controller and analyze its properties. 
An example for system \eqref{eq: model dynamics} was employed from \cite{DixonDNN} as 
$\dot x = f(x) + u$ with
\begin{equation}
    f(x)
     = 
    \begin{bmatrix}
        x_1x_2\tanh(x_2)+\text{sech}(x_1)\\
        \text{sech}^2 (x_1+x_2) -\text{sech}^2(x_2)
    \end{bmatrix},
    \label{eq: sim dynamics}
\end{equation}
where $x=[x_1,x_2]^\top  $, $u=[u_1,u_2]^\top$. 
The initial state value was $x(0)=[1,2]^\top  $, and the desired trajectory was $x_d(t) = [\sin(2t),-\cos(t)]^\top  $. 

Six controllers are compared: the proposed controller with default parameters (CNN1), four variations thereof (CNN2 to CNN5), and the comparison controller (DNN). CNN1 had the following control gain and design parameters: $k_s=1, \rho=10^5$, and $A_c=-10I$, where $I = [1,0;0,1]$. The CNN architecture of CNN1 was designed as follows: $k_f=2, k_c=1, T_s=0.1, n_0 = 10$, $\alpha_1 = 100$, $\alpha_2=0.01$; each FCL had $8$ nodes; the first CVL had $2$ filters whose dimensions were $(5\times 6)$; the second layer had $2$ filters whose dimensions were $(3\times2)$; the input matrix $X$ had the dimensions of $(10\times 6)$; all trainable variables were initialized from a uniform distribution in the interval $(-0.1,0.1)$ and not pretrained; and $\tanh(\cdot)$ was selected as the FCLs and CVLs's activation function.

CNN2 used a smaller stacking time value (i.e. $T_s=0.01$) compared to CNN1; CNN3 replaced the adaptation law of CNN1 with one from \cite{DixonDNN} (i.e., $A_c=-I, \rho = 0$); CNN4 used a larger damping factor value (i.e. $\rho=5\cdot10^5$) compared to CNN1; CNN5 employed eigenvalues that are more negative for $A_c$ (i.e. $A_c=-50I)$ compared to those used by CNN1; and DNN was defined by excluding the CVLs from the proposed CNN architecture. 
DNN had $4$ FCLs with $8,8,8$ and $4$ nodes 
using $\alpha_1\tanh(\xi(t)/\alpha_2)$ as the input vector. The other parameters were the same as CNN1. 
The numbers of nodes in layers were chosen as described above so that the total number of weights of DNN (250) was almost similar to those of CNN1 (244).

The simulation results are presented in Fig.~\ref{fig:sim}, with the tracking errors compared in Table~\ref{table: error norm}, where $e_{1(2)}$ and $\epsilon_{1(2)}$ denote the tracking error and root mean square error (RMSE) of $x_{1(2)}$, respectively. Notably, CNN1, CNN2, and CNN4 showed almost asymptotic convergence, which confirmed the main finding of this study. A detailed comparison is performed from two perspectives in the subsections below.

\begin{figure}%
    \centering
        \subfigure[]{\includegraphics[width=0.45\linewidth]{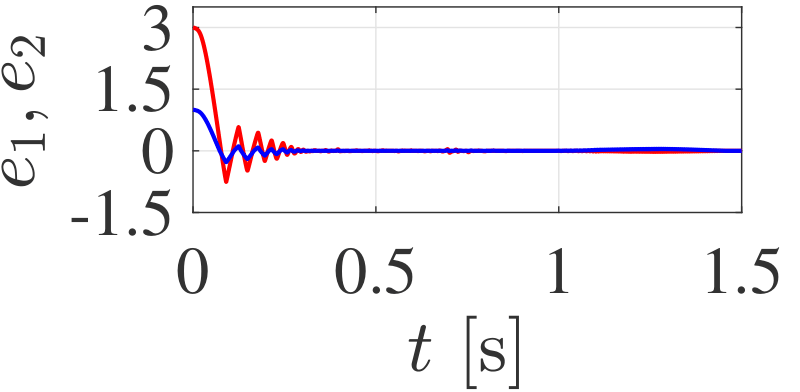}%
        \label{fig:CNN1}}
    \hfill
        \subfigure[]{\includegraphics[width=0.45\linewidth]{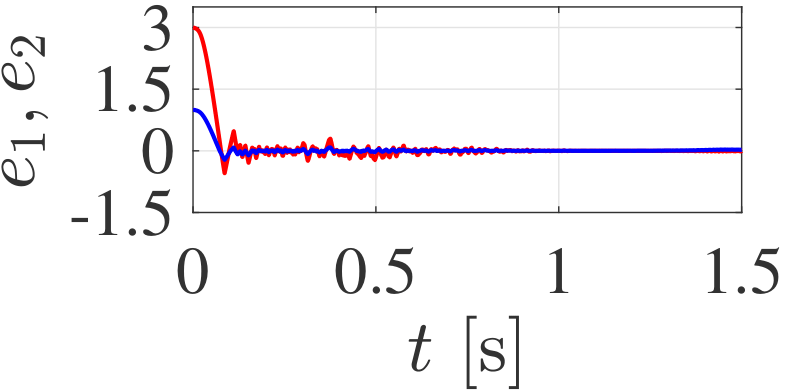}%
        \label{fig:CNN2}}
    \vfill
        \subfigure[]{\includegraphics[width=0.45\linewidth]{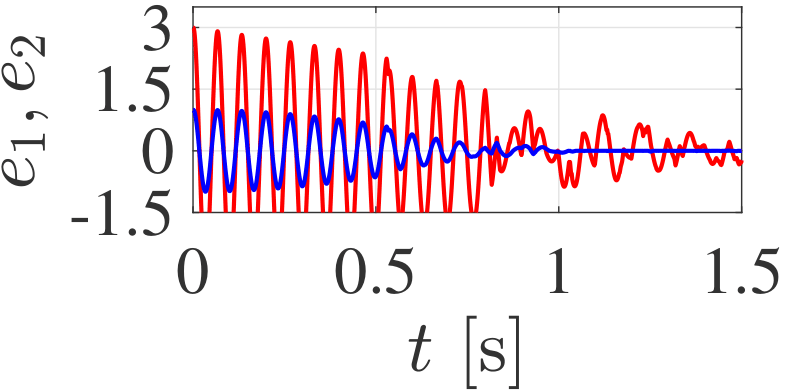}%
        \label{fig:CNN3}}
    \hfill
        \subfigure[]{\includegraphics[width=0.45\linewidth]{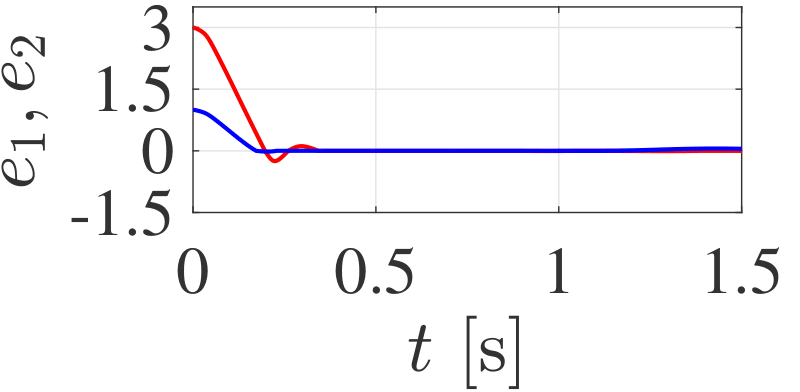}%
        \label{fig:CNN4}}
    \vfill
        \subfigure[]{\includegraphics[width=0.45\linewidth]{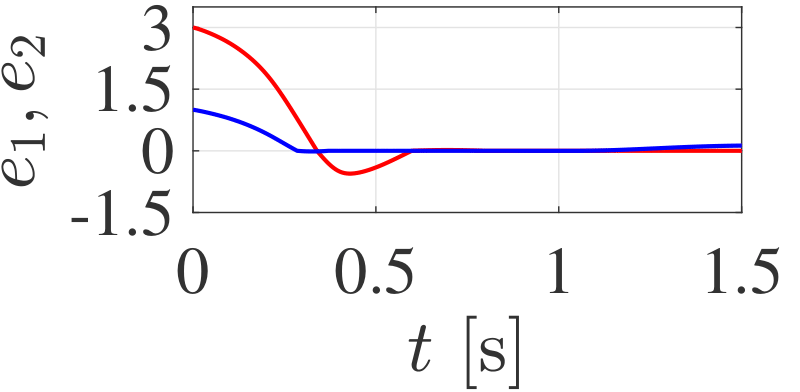}%
        \label{fig:CNN5}}
    \hfill
        \subfigure[]{\includegraphics[width=0.45\linewidth]{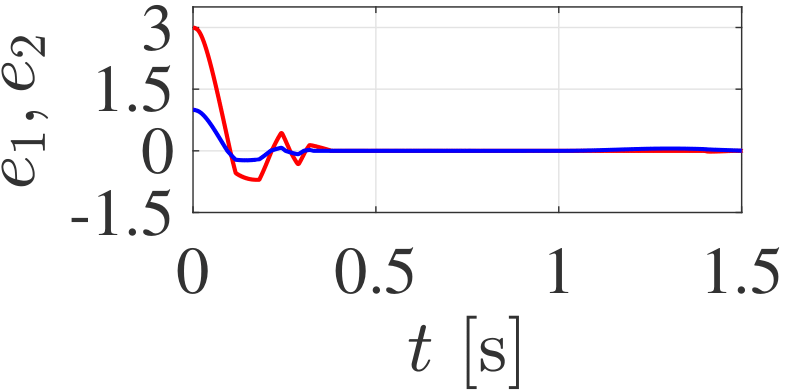}%
        \label{fig:DNN}}
    \caption{Tracking errors of (a) CNN1, (b) CNN2, (c) CNN3, (d) CNN4, (e) CNN5, and (f) DNN. The blue and red solid lines represent $e_1$ and $e_2$, respectively.}
    \label{fig:sim}
\end{figure}

\subsection{Effects of Design Parameters}
\begin{itemize}
\item Stacking time $T_s$: CNN2, with a smaller $T_s$, regulated the tracking errors faster than did CNN1 but provided noisy responses. This is because a small value of $T_s$ provides fine but shortsighted dynamic information, as mentioned in Remark~\ref{remark:stack param}.
\item Desinger matrix $A_c$ and damping factor $\rho$: CNN3, with a simpler adaptation law ($A_c = -I$ and $\rho = 0$), showed slow and oscillating convergence. In contrast, other CNNs using more negative $A_c$ and nonzero $\rho$ showed faster convergence than did CNN3 and almost no oscillating convergence. The larger $\rho$ was, the slower but more stable the convergence was (see CNN4 vs. CNN1 and CNN2). However, a more negative $A_c$ did not guarantee faster weight's convergence (see CNN5), as mentioned in Remark~\ref{remark:2}.
\end{itemize}

\begin{table}%
    \renewcommand{\arraystretch}{1.3}
    \caption{Quantitative Comparison of Tracking RMSE}
    \label{table: error norm}
    \centering
    \begin{tabular}{|c||c|c|c|}
    \hline
    & \bfseries CNN1 & \bfseries CNN2 & \bfseries CNN3 \\
    & (Default) & (Smaller $T_s$) & (Existing adaptation law) \\
    \hline 
    $\epsilon_1$ & 0.0397 & 0.0384 & 0.2160 \\
    \hline
    $\epsilon_2$ & 0.3752 & 0.3680 & 2.4030 \\
    \hline    
    & \bfseries CNN4 & \bfseries CNN5 & \bfseries DNN \\
    & (Larger $\rho$) & (More negative $A_c$) & (No CVLs) \\
    \hline 
    $\epsilon_1$ & 0.0716 & 0.1291 & 0.0490 \\
    \hline
    $\epsilon_2$ & 0.7524 & 1.5446 & 0.4757 \\
    \hline
    \end{tabular}
    \label{table: erro}
\end{table}

\subsection{CNN vs. DNN}

The strength of the proposed CNN-based controller is derived from utilizing image data containing system historical information, which is useful for approximating a dynamic function of $f_c(x)-\dot x_d+A_cx_d$ in \eqref{eq: desired input}. However, the DNN-based controller uses only the current observation and is considered a state feedback controller that has a limited capability to approximate the dynamical function. %

The difference in performance between DNN and CNN1 was not significant in the simulation results provided in Fig.~\ref{fig:sim} and Table~\ref{table: erro}, although DNN had slightly larger errors than did CNN1. However, the difference can be highlighted by a case study where the system changed suddenly at $t=3$ from \eqref{eq: sim dynamics} to $\dot x = f(x)+g(x)+u$, with 
\begin{equation}
    g(x) = 
	\begin{bmatrix}
		2x_1^2x_2 + 2\sin(t)+20\\
		2x_2^2\tanh(x_1)+2\cos({1\over2}t)+20
	\end{bmatrix}.
\end{equation}
The case study results, provided in Fig.~\ref{fig: FCN vs CNN}, demonstrated that CNN1 learned the new desired control policy (which depended on the time-dependent functions) faster than DNN did.

\begin{figure}[!t]
    \centering
    \includegraphics[width=\linewidth]{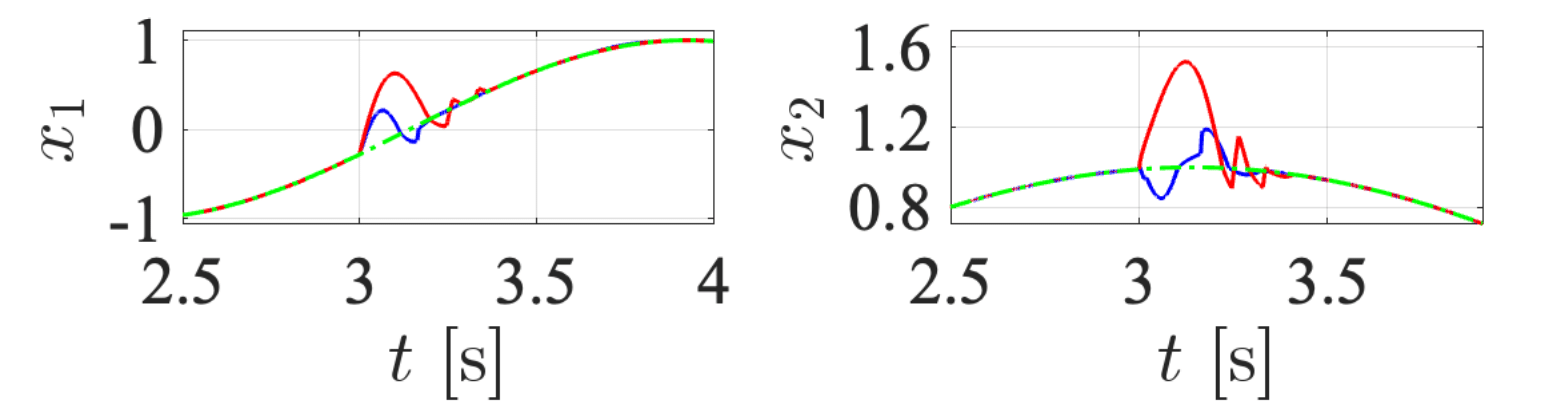}
    \caption{Control results of CNN1 (blue solid line) and DNN (red solid line) under a sudden change in the system at $3$ s. The green dash-dotted line denotes the desired trajectory.}
    \label{fig: FCN vs CNN}
\end{figure}

\section{Conclusion and Future Work}
This letter presents a CNN-based end-to-end adaptive controller for tracking control of uncertain control-affine nonlinear systems. The main contributions are threefold: 1) to present an analytical expression of the CNN architecture, 2) to derive an adaptation law based on gradient descent optimization for the CNN architecture, and 3) to prove the stability of the adaptation law, based on Lyapunov analysis. A simulation study demonstrated that the proposed controller learned the desired control policy and provided asymptotic convergence without using pretrained data. 
Future work would be development of other CNN variations-based controller with stability guarantee.

\section*{APPENDIX}

\subsection{CNN Operator}\label{A. CNN Operator}

The CNN operator $O(\cdot)$ can be represented as 
\begin{equation}
    O(X,\Omega,B)_{(i,j)} = W^{(j)} \odot \text{row}_{(i:i+p-1)} (X) + B_{(j)},
\end{equation}
where $X\in\mathbb{R}^{n\times m}$, $\Omega \in\mathbb{R}^{p\times m\times q}$, and $B\in\mathbb{R}^{q}$ denote the input matrix, filter set and bias vector, respectively.
The filter set $\Omega=\{W^{(1)},\dots,W^{(q)}\}$ has $q$ filters $W^{(i)}\in\mathbb{R}^{p\times m},\ \forall i\in[1,\dots, q]$. 
Note that the operator's output is also bounded if $X$, $\Omega$, and $B$ are bounded.

\subsection{Jacobians with respect to Convolutional Layers}\label{A. CNN Jacobian}

The Jacobian of $\Phi$ with respect to the concatenate layer is represented as
\begin{equation}
    {\partial \Phi\over\partial C(\Phi^C_{k_c})} = 
    (\overset{\curvearrowleft}\prod^{k_f}_{l=1} 
    {\hat V}^\top  _l {\hat\phi}_l^{F'}  ) V_0^\top  .
    \label{eq. dPhi/dC}
\end{equation}
The Jacobian of $\Phi_i$ with respect to the CVLs' output is
\begin{equation}
    {\partial \Phi_i\over\partial \Phi^C_{k_c}} = 
    \text{reshape}
    \bigg(
    {\partial \Phi_i\over\partial  C(\Phi^C_{k_c})}
    _{(1:n_{k_c}m_{k_c})}
    ,n_{k_c},m_{k_c}
    \bigg).
    \label{eq, dC/dPhi to dPhidPhi}    
\end{equation}
The Jacobians of $\Phi_i$ with respect to the $j_c\textsuperscript{th}$ activation function of the CVLs can be represented as
\begin{equation}
    \begin{aligned}
        {\partial \Phi_i\over\partial \phi^{C}_{j_c}} =&
        \sum_{l_i=1}^{n_{j_c+1}} \sum_{l_j=1}^{m_{j_c+1}} 
        \bigg(
            { \partial \Phi_i\over\partial {\Phi^C_{j_c}}_{(l_i,l_j)} }
            {\partial {\Phi^C_{j_c}}_{(l_i,l_j)}\over\partial \phi^{C}_{j_c}}
        \bigg)
        \\=&
        {\partial\Phi_i\over\partial {\Phi^C_{j_c}}_{(1,1)}}
        \begin{bmatrix}
            W_{{j_c}}^{(1)}\\ 0_{{(n_{{j_c}}}-p_{{j_c}})\times m_{{j_c}}}
        \end{bmatrix}
        \\&
        +
        {\partial\Phi_i\over\partial {\Phi^C_{j_c}}_{(2,1)}}
        \begin{bmatrix}
            0_{(n_{{j_c}}-p_{{j_c}})\times 1}\\ 
            W_{{j_c}}^{(1)}\\
            0_{({n_{{j_c}}}-p_{{j_c}})\times (m_{{j_c}}-1)}
        \end{bmatrix}
        \\
        &+\cdots
        +{\partial\Phi_i\over\partial {\Phi^C_{j_c}}_{(n_{j_c+1}, m_{j_c+1})}}
        \begin{bmatrix}
            0_{{(n_{{j_c}}}-p_{{j_c}})\times m_{{j_c}}}\\ \\ 
            W_{{j_c}}^{(q_{j_c})}
        \end{bmatrix}
        \end{aligned}
\end{equation}
for all $j_c\in[1,\dots,k_c]$, where $\partial \Phi_i/\partial \Phi^C_{j_c}$ is the Jacobians of $\Phi_i$ with respect to $\Phi^C_{j_c}$ and represented as
\begin{equation}
    {\partial \Phi_i\over\partial {\Phi^C_{j_c}}} = {\partial {\Phi_i}\over\partial \phi^{C}_{j_c+1}}\odot{\phi^{C}_{j_c+1}}'({\Phi^C_{j_c}})
\end{equation}
for all $j_c\in[0,\dots,k_c]$.
Finally, the Jacobians of $\Phi_i$ with respect to the $l_k\textsuperscript{th}$ filters and bias vectors of the $j\textsuperscript{th}$ CVL can be obtained as
\begin{equation}
    \begin{aligned}
        {\partial \Phi_i\over\partial W^{l_k}_{j_c}} 
        &=
        \sum_{l_i=1}^{n_{j_c+1}}
        \bigg(
        {\partial \Phi_i\over\partial {\Phi^C_{j_c}}_{(l_i,l_k)}} 
            \text{row}_{(l_i:p_{j_c}+l_i-1)}(\phi^C_{j_c})
        \bigg)
        \\
         &\hspace{90pt}\text{with }j_c\in[0,\dots,k_c],\hfilneg\\        
        {\partial \Phi_i\over\partial {B_{j_c}}_{(l_k)}} &=
        \sum_{l_i=1}^{n_{j_c+1}} 
        \bigg(
            {\partial \Phi_i\over\partial {\Phi^C_{j_c}}_{(l_i,l_k)}} \cdot 1
        \bigg)
        \\
         &\hspace{90pt}\text{with }j_c\in[0,\dots,k_c],\hfilneg
    \end{aligned}
    \label{eq. dPhidOm}
\end{equation}
for all $j_c\in[0,\dots,k_c]$ and $l_k\in[1,\cdots,q_{j_c}]$.

\bibliographystyle{ieeetr}
\bibliography{refs}

\end{document}